\documentclass[reqno,b5paper]{amsart}
\usepackage{amsmath,amsthm,amssymb}
\usepackage{amsfonts}
\usepackage{mathrsfs,cite}
\usepackage[usenames,dvipsnames]{pstricks}
\usepackage{epsfig}

\newcommand{\al}{\alpha}

\newcommand{\I}{{\mathbb I}}
\newcommand{\R}{{\mathbb R}}
\newcommand{\Z}{{\mathbb Z}}
\newcommand{\N}{{\mathbb N}}

\newtheorem{theorem}{Theorem}[section]

\theoremstyle{definition}

\theoremstyle{remark}
\newtheorem{remark}[theorem]{Remark}

\numberwithin{equation}{section}

\begin{document}

\title[Notes on a parameter switching method]
 {Note on a parameter switching method for nonlinear ODEs}
\author[Marius-F. Danca \and Michal Fe\v{c}kan]{Marius-F. Danca** \and Michal Fe\v{c}kan* (alphabetic order)}
\newcommand{\acr}{\newline\indent}
\address{\llap{*\,}
    Department of Mathematical Analysis and Numerical Mathematics\acr
    Comenius University\acr
    Mlynsk\'a dolina\acr
    842 48 Bratislava\acr
    SLOVAKIA\acr
    and\acr
    Mathematical Institute\acr
    Slovak Academy of Sciences\acr
    \v Stef\'anikova 49\acr
    814 73 Bratislava\acr
    SLOVAKIA}
\email{Michal.Feckan@fmph.uniba.sk}
\address{\llap{**\,}
     Department of Mathematics and Computer Science\acr
    Emanuel University\acr
    Str. Nufarului nr. 87, \acr
    410597 Oradea, \acr
    Romania\acr
    and \acr
    Romanian Institute for Science and Technology\acr
    Str. Cireºilor 29\acr
    400487 Cluj-Napoca\acr
    ROMANIA}
\email{danca@rist.ro}
\thanks{M. Fe\v ckan is partially supported by Grants VEGA-MS 1/0507/11, VEGA-SAV 2/0029/13 and APVV-0134-1}

\subjclass[2010]{Primary 34K28; Secondary 34C29, 37B25, 34H10}

\keywords{Numerical approximation of solutions, aeraging method, Lyapunov function, chaos control, anticontrol}

\begin{abstract}
In this paper we study analytically a parameter switching (PS) algorithm applied to a class of systems of ODE, depending on a single real parameter. The algorithm allows the numerical approximation of any solution of the underlying system by simple periodical switches of the control parameter. Near a general approach of the convergence of the PS algorithm, some dissipative properties are investigated and the dynamical behavior of solutions is investigated with the Lyapunov function method. A numerical example is presented.
\end{abstract}

\maketitle

\section{Introduction}\label{intro}
In \cite{D1} it was proved that the Parameter Switching (PS) algorithm, applied to a class of Initial Values Problems (IVP) modeling a great majority of continuous nonlinear and autonomous dynamical systems depending to a single real control parameter, allows to approximate any desired solution, while in \cite{D3} several applications are presented. By choosing a finite set of parameters values, PS switches in some deterministic (periodic) way the control parameter within the chosen set, for relative short time subintervals, while the underlying IVP is numerical integrated. The obtained ``switched'' solution will approximate the ``averaged'' solution obtained for the parameter replaced with the averaged of the switched values. As verified numerically, (see e.g. \cite{D2}), the switchings can be implemented even in some random manner but the convergence proof is much more complicated.

The PS algorithm applies to systems modeled by the following autonomous IVP
\begin{equation}\label{e0}
\begin{gathered}
\dot x(t)=f(x(t))+pAx(t),\quad t\in I=[0,T], \quad x(0)=x_0,
\end{gathered}
\end{equation}
\noindent for $T>0$, $x_0\in \R^n$, $p\in \R$, $A\in L(\R^n)$ and $f : \R^n\to\R^n$ a nonlinear function.

In order to ensure the uniqueness of solution, the following assumption is considered

\textbf{H1}
$f$ satisfies the usual Lipschitz condition

\begin{equation}\label{lip}
|f(y_1)-f(y_2)|\le L|y_1-y_2|\forall y_{1,2}\in\R^n,
\end{equation}

\noindent for some $L>0$.

By applying the PS algorithm to a system modeled by the IVP (\ref{e0}), it is possible to approximate any desired solution and, consequently, any attractor of the underlying system, by simple switches of the control parameter.

The great majority of known dynamical systems, such as: Lorenz, Chua, R\"{o}ssler, Chen, Lotka-Volterra, to name just a few, are modeled by the IVP (\ref{e0}).

The PS algorithm applies not only to integer-order systems, but also for fractional-order systems \cite{D2}, discrete real systems \cite{D4} and complex systems \cite{M}.

In this work, near a general approach of the convergence of the PS algorithm, some dissipative properties are discussed by means of the Lyapunov function method.

The paper is organized as follows: Section \ref{sec2} presents the PS algorithm and his numerical implementation, in Section \ref{sec3} is presented a general approach of the convergence of the PS algorithm, Section \ref{sec4} gives an estimation between the average and switched numerical solutions, while in Section \ref{sec5} a Lyapunov approach is presented. The paper ends with a conclusions section.

\section{PS algorithm}\label{sec2}

Let us partition the integration interval $I=\bigcup_j (\bigcup_{i=1}^{N} I_{i,j})$, $j\geq1$ (see the sketch in Figure \ref{fig1}, where the particular case $N=3$ has been considered), and let also denote by
$$
\mathcal{P}=\{p_1,p_2,...,p_N\}\subset \R,~~N\geq2,
$$
the switching values set.

\begin{figure*}
  \includegraphics[clip,width=0.75\textwidth]{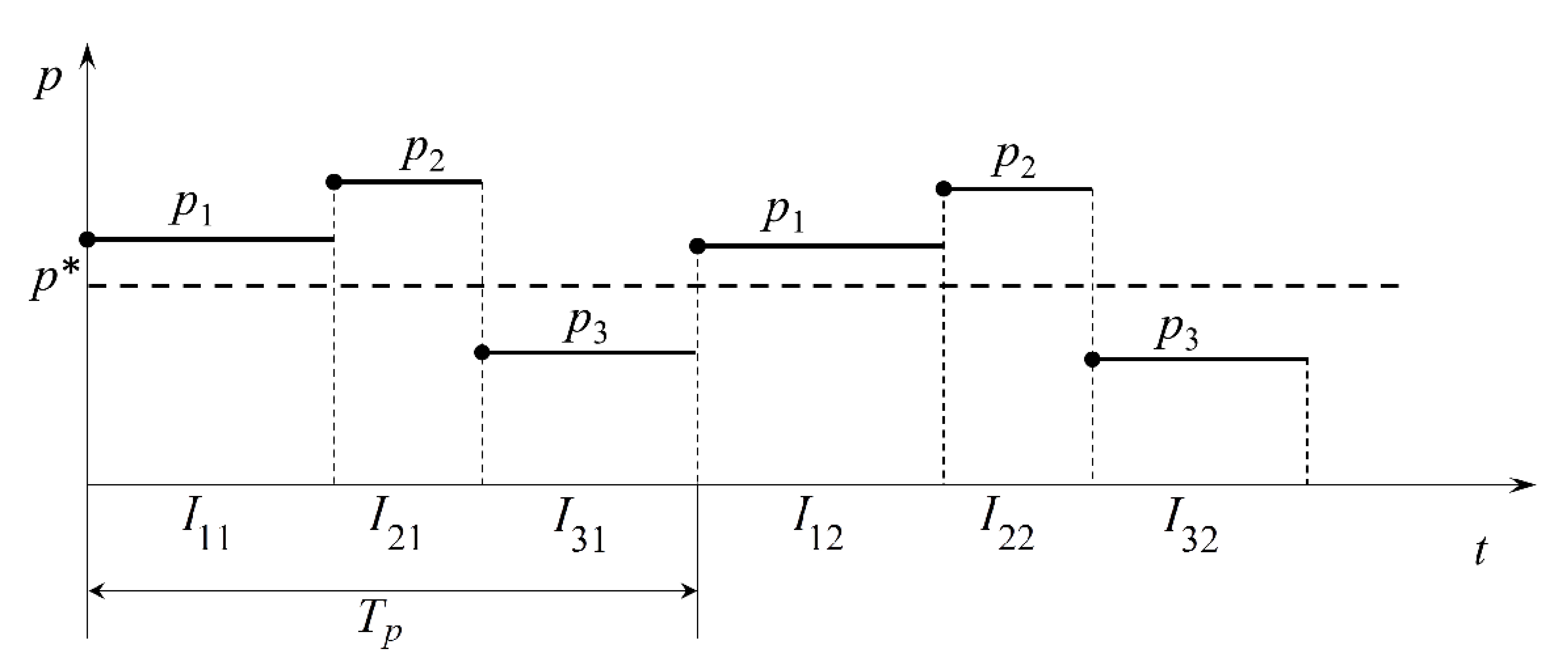}
%\vspace*{5cm}       % Give the correct figure height in cm
\caption{Sketch of the time interval $I$ partition, for the case $N=3$. }
\label{fig1}       % Give a unique label
\end{figure*}

$p$ will be considered a periodic piece-wise constant function $p:I\rightarrow \mathcal{P}$, $p(t)=p_i$, for $t\in I_{i,j}\cap I$, $1\leq i\leq N$ for every $j\geq1$. In this way, $p$ will be switched in every subinterval $I_{i,j}$ (Figure \ref{fig1}).

\vspace{2mm}
Due to the $p$ periodicity, for the sake of simplicity, let us drop next the index $j$ unless necessary.
\vspace{2mm}

In order to specify $p$ more concretely, let $h>0$ be the switching step size. Then, the subintervals $I_i$ can be expressed as follows: $I_{i}=[M_{i-1}h,M_{i}h)$, for $1\le i\le N$, with $M_{0}=0$, $M_{i}:=\sum_{k=1}^{i}m_k$, $m_i\in \N^*$ being the ``weights'' of $p_i$ in the subinterval $I_i$. Thus, $p(t)=p_h(t)=p_i$ for $t\in [M_{i-1}h,M_{i}h)$, $1 \leq i \leq N$, with period $T_p:=hM_{N}$. We suppose $T>T_p$. Thus, in the above notation we have $I_{i,j}=I_{i,j-1}+(j-1)T_p$. For example, in Figure \ref{fig1}, $I_{3,2}=I_{3,1}+T_p$.

Let denote the ``weighted average'' of the values of $\mathcal{P}$ by

\begin{equation}\label{p}
p^*:=\frac{\sum_{i=1}^Np_im_i}{\sum_{i=1}^Nm_i},
\end{equation}

\noindent Then, the switched solution, obtained with the PS algorithm by switching $p$ to $p_i$ in each subinterval $I_i$, $1\leq i \leq N$, while the underlying IVP is numerically integrated with some fixed step size $h$, will tend to the averaged solution obtained for $p$ replaced with $p^*$ \cite{D1}. Precisely, these two solutions are $O(h)$ near on $[0,T]$.

\begin{remark}
The simplicity of the PS algorithm resides in the linear appearance of $p$ in the term $pAx(t)$ in (\ref{e0}).
\end{remark}

To implement numerically the PS algorithm, the only we need is to choose some numerical method with the single step-size $h$ to integrate the underlying IVP. Thus, let us suppose one intend to approximate with PS the solution corresponding to some value $p^*$ \footnote {Due to the supposed uniqueness, to different $p$, different solutions.}. Then, we have to find a set $\mathcal{P}$ and the corresponding weights $m$, such that (\ref{p}) gives the searched value $p^*$ (details on the numerical implementation can be found in \cite{D1} or \cite{D3}).

For example, suppose we intend to approximate a stable cycle of the Lorenz system

\begin{equation*}
\begin{array}
[c]{cl}
\overset{\cdot}{x}_{1}= & a(x_{2}-x_{1}),\\
\overset{\cdot}{x}_{2}= & x_{1}(r-x_{3})-x_{2},\\
\overset{\cdot}{x}_{3}= & x_{1}x_{2}-cx_{3},
\end{array}
\label{lorenz}
\end{equation*}

\noindent corresponding to $p:=r=155$ and to the usually values $a=10$ and $c=8/3$. In this case, $f(x)$ and $A$ in (\ref{e0}) are

\[
f(x)=\left(
\begin{array}{c}
a(x_{2}-x_{1}) \\
-x_{1}x_{3}-x_{2} \\
x_{1}x_{2}-cx_{3}%
\end{array}%
\right) ,~~A=\left(
\begin{array}{ccc}
0 & 0 & 0 \\
1 & 0 & 0 \\
0 & 0 & 0%
\end{array}%
\right) .
\]

\noindent By using e.g. $\mathcal{P}=\{150,152, 168\}$, one of the possible choices for weights to obtain $p^*=155$ in (\ref{p}), is $m_1=2, m_2=1$ and $m_3=1$, i.e. $p^*=(2\times150+1\times152+1\times168)/(2+1+1)=155$. Next, by applying the PS algorithm via the standard RK method (utilized here) with $h=0.001$, after $T=10sec$ one obtains a good match between the two solutions: averaged solution (in blue) and switched (in red) (Figure \ref{fig2}).

\begin{figure*}
  \includegraphics[clip,width=0.65\textwidth]{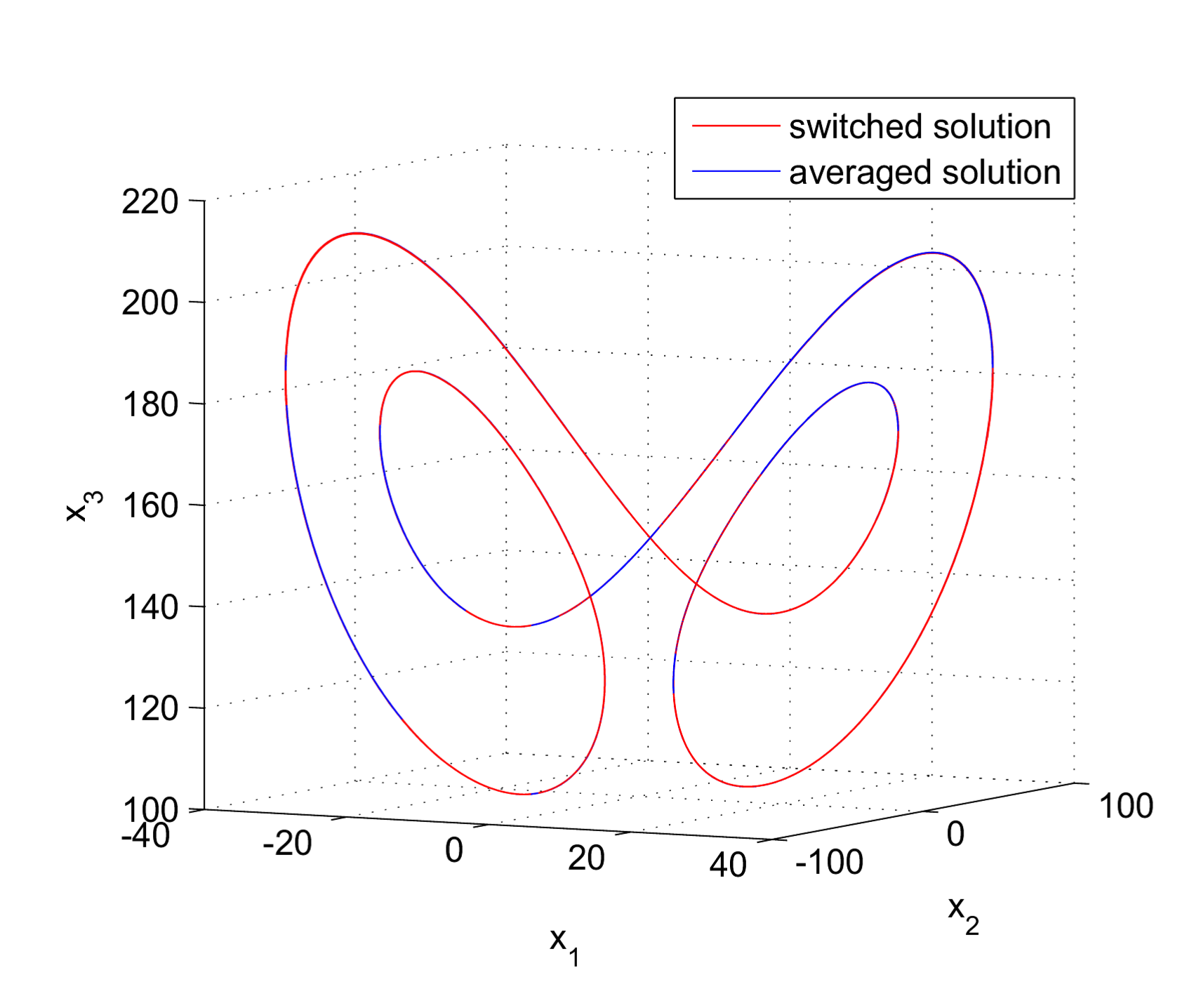}
%\vspace*{5cm}       % Give the correct figure height in cm
\caption{Phase plot of the stable cycle of the Lorenz system corresponding to $p=155$, approximated with the PS algorithm by switching the values $\mathcal{P}=\{150,152, 168\}$. The averaged (blue) and the switched (red) solutions are over-plotted.}
\label{fig2}       % Give a unique label
\end{figure*}

In all images, the transients have been neglected.

The same stable cycle (corresponding to $p=155$) can be obtained with other sets $\mathcal{P}$ and weights $m$ verifying (\ref{p}). Thus, this stable movement can be approximated, for example, by switching the following seven values: $\mathcal{P}=\{125,130,140,156,160,168,200\}$ with weights $m_1=2, m_2=3, m_3=2, m_4=2, m_5=3, m_6=1, m_7=3$. Again, these values gives $p^*=155$. This time, due to the inherent numerical errors (see \cite{D1}), to obtain a good fitting between the two curves, we had to choose a smaller step size $h=0.0001$ (Figure \ref{fig3}).

Another stable cycle, corresponding to $p=93$ is approximated in Figure \ref{fig4}, by using $\mathcal{P}=\{89,92,95\}$ and $m_1=1, m_2=2, m_3=3$ and integration step size $h=0.001$.

\begin{figure*}
  \includegraphics[clip,width=0.75\textwidth]{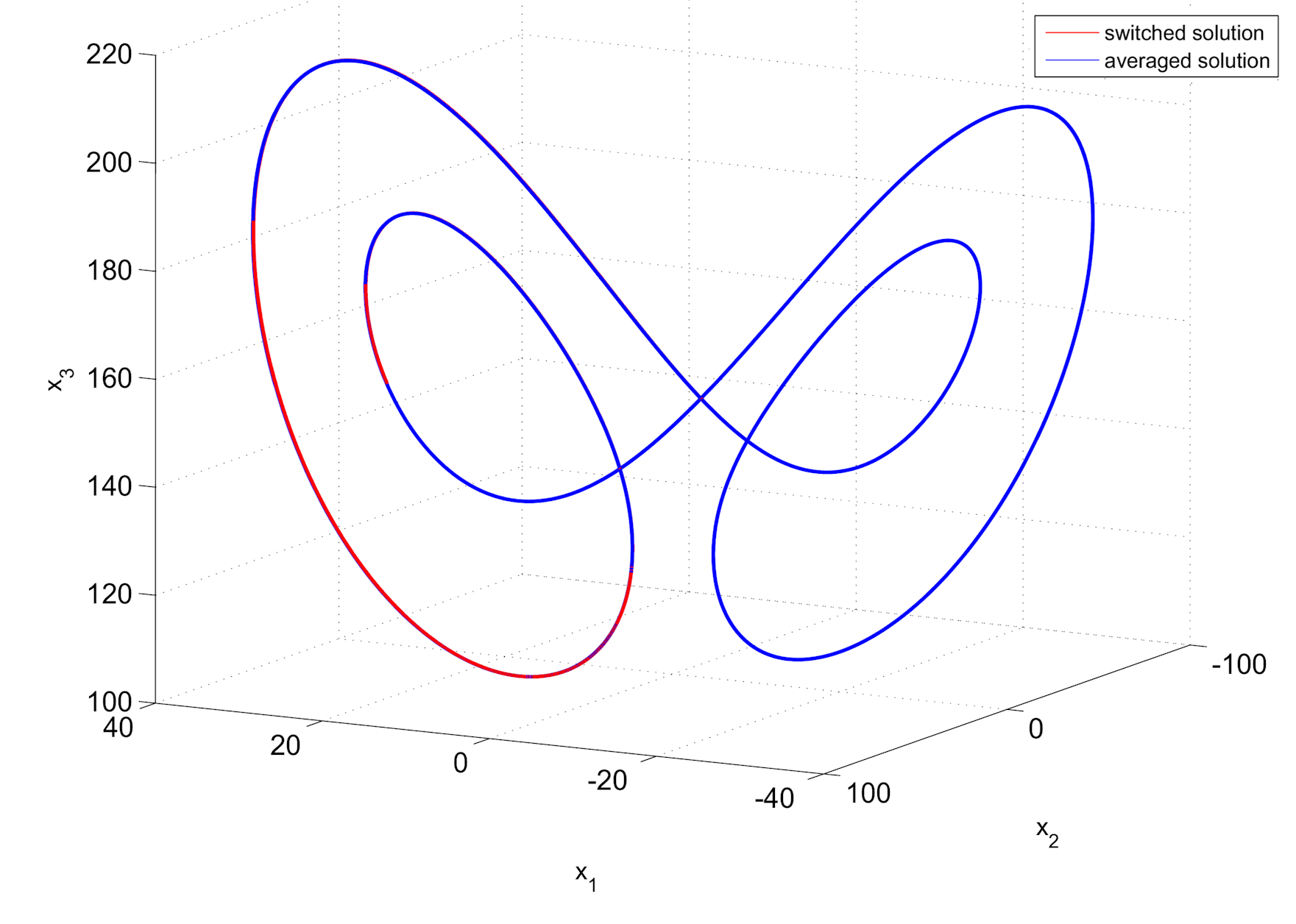}
%\vspace*{5cm}       % Give the correct figure height in cm
\caption{Phase plot of the stable cycle of the Lorenz system corresponding to $p=155$, approximated with the PS algorithm by switching the values $\mathcal{P}=\{125,130,140,156,160,168,200\}$.}
\label{fig3}       % Give a unique label
\end{figure*}

\begin{remark}\label {rem}
\item [(i)] PS algorithm is useful when, due to some objective reasons, some desired targeted value $p^*$ cannot be accessed directly.

\item [(ii)] Since, as we shall see in the next sections, with the PS algorithm any solution of the IVP (\ref{e0}) can be numerically approximated, when the obtained (switched) solution gives rise to a stable cycle, the algorithm can be considered as a chaos control-like method while, when chaotic motions are approximated, it can be considered an anticontrol-like method (see e.g. \cite{D3}). Compared to the known control methods, such as OGY-like schemes, where the unstable periodic orbits are ''forced'' to become stable, with the PS algorithm one approximate any desirable stable orbit.

\item[(iii)] The size of $h$ is considered to be stated implicitly by the utilized convergent numerical method for ODEs (here the standard RK, see e.g. \cite{SH}).
\end{remark}

\begin{figure*}
  \includegraphics[clip,width=0.75\textwidth]{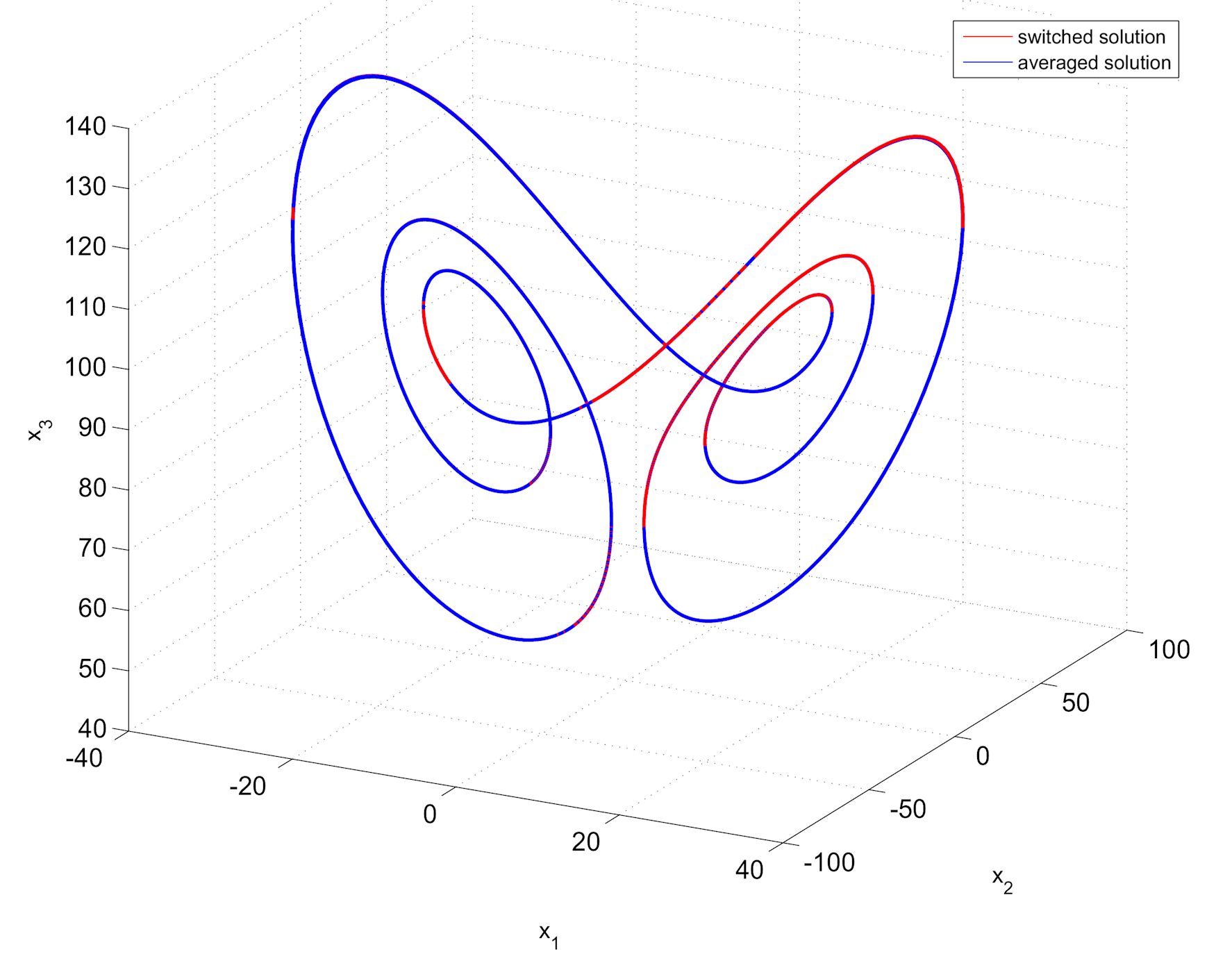}
%\vspace*{5cm}       % Give the correct figure height in cm
\caption{Phase plot of the stable cycle of the Lorenz system corresponding to $p=93$, approximated with the PS algorithm by switching the values $\mathcal{P}=\{89,92,95\}$.}
\label{fig4}       % Give a unique label
\end{figure*}

\section{Convergence of PS algorithm}\label{sec3}

 In this Section we prove the PS algorithm convergence.

Let $p_h(t):=P(t/h)$ for any $h>0$, $h$ being considered as a small parameter (see also Remark \ref{rem} (iii)). Then, the switching equation has the following form
\begin{equation}\label{e1}
\begin{gathered}
\dot x(t)=f(x(t))+p_h(t)Ax(t),\quad t\in I=[0,T],\quad x(0)=x_0,
\end{gathered}
\end{equation}
The averaged equation of (\ref{e1}) (obtained for $p$ replaced with $p^*$ given by (\ref{p})), is
\begin{equation}\label{e2}
\begin{gathered}
\dot {\bar x}(t)=f(\bar x(t))+p^*A\bar x(t),\quad t\in I=[0,T],\quad \bar x(0)=\bar x_0.
\end{gathered}
\end{equation}
Under the assumption \textbf{H1}, the convergence of the PS algorithm is given by the following theorem

\begin{theorem}
Let $\|\cdot\|_0$ be the maximum norm on $C([0,T],\R^n)$,  i.e., $\|\bar x\|_0:=\max_{t\in[0,T]}|\bar x(t)|$ for $\bar x$. Under the above assumptions, it holds
\begin{equation}\label{est1}
\begin{gathered}
|x(t)-\bar x(t)|\le ( |x_0-\bar x_0|+h\|A\|\|\bar x\|_0K)
\times e^{(L+\|P\|_0\|A\|)T},
\end{gathered}
\end{equation}
for all $t\in [0,T]$, where
$$
K:=\max_{t\in[0,M_N]}\left|\int_0^{t}(P(s)-p^*)ds\right|.
$$
\end{theorem}

\begin{proof}

From \eqref{e1} and \eqref{e2} we get
$$
\begin{gathered}
|x(t)-\bar x(t)|\le |x_0-\bar x_0|+ L\int_0^t|x(s)-\bar x(s)|ds
+\left|\int_0^t(p_h(s)-p^*)ds\right|\|A\|\|\bar x\|_0\\
+\|P\|_0\|A\|\int_0^t|x(s)-\bar x(s)|ds
= |x_0-\bar x_0|+\|A\|\|\bar x\|_0\left|\int_0^t(p_h(s)-p^*)ds\right|\\
+(L+\|P\|_0\|A\|)\int_0^t|x(s)-\bar x(s)|ds.
\end{gathered}
$$
Note that
\begin{equation*}
\max_{t\in[0,T]}|p_h(t)|\le \|P\|_0=\max_{t\in[0,M_N]}|P(t)|=\max|m_i|.
\end{equation*}
Next it holds
$$
\int_0^t(p_h(s)-p^*)ds=h\int_0^{t/h}(P(s)-p^*)ds.
$$
Since $\int_0^{t}(P(s)-q)ds$ is $M_N$-periodic, we have
$$
\max_{t\in[0,T]}\left|\int_0^{t/h}(P(s)-p^*)ds\right|\le K.
$$
Note $K$ is independent of $T$. Hence
$$
\begin{gathered}
|x(t)-\bar x(t)|= |x_0-\bar x_0|+h\|A\|\|\bar x\|_0K
+(L+\|P\|_0\|A\|)\int_0^t|x(s)-\bar x(s)|ds.
\end{gathered}
$$
Finally, by Gronwall inequality \cite{HH}, we obtain \eqref{est1}. The proof is finished.
\end{proof}
\begin{remark}
This proof is more general than the proof presented in \cite{Mao}, where the convergence is obtained via the averaging method (see  \cite{SVM}).
\end{remark}

\section{Numerical approximation estimates}\label{sec4}

Next, by using numerical approximation estimates, another generalized proof of the convergence of PS algorithm is presented.

Let us consider again the switched and averaged equation (\ref{e1}) and (\ref{e2}) respectively.

For each
$$
t\in[M_ih+kT_p,M_{i+1}h+kT_p)\cap[0,T],\quad k\in \Z,
$$
the differential equation \eqref{e1} is actually an autonomous ODE
\begin{equation}\label{e3}
\dot x(t)=f(x(t))+pAx(t),
\end{equation}
with $p=p_i$. Let $\varphi_p(t,x)$ be the flow of \eqref{e3}. Then it holds
$$
\begin{gathered}
|\varphi_p(t,x_1)-\varphi_p(t,x_2)|\le |x_1-x_2|
+ (L+|p|\|A\|)\int_0^t|\varphi_p(s,x_1)-\varphi_p(s,x_2)|ds,
\end{gathered}
$$
and by Gronwall inequality \cite{HH} we obtain
\begin{equation}\label{e4}
|\varphi_p(t,x_1)-\varphi_p(t,x_2)|\le |x_1-x_2| e^{(L+|p|\|A\|)t}.
\end{equation}

For some discretization scheme  $\psi\in C([0,h_0]\times \R^n, \R^n)$, with $h_0\in (0,1)$ of \eqref{e3}, we suppose \cite{HLW,SH} that
\begin{equation}\label{order}
|\psi_p(h,x)-\varphi_p(h,x)|\le Mh^{r+1}\, \forall (h,x)\in [0,h_0]\times \R^n,
\end{equation}
for some $M>0$ and $r\in\N$.

Let us consider now the  numerical approximation sequence of the switched equation (\ref{e1}) $u_0:=x_0$, and by induction: $u_{j+1}:=\psi_p(h,u_j)$, with $p=p_i$ and $jh\in[M_ih+kT_p,M_{i+1}h+kT_p)\cap[0,T]$.

Similarly, we consider the sequences corresponding to the averaged equation (\ref{e2}) $\bar v_j=\bar x(jh),\bar v_j^*:=\psi_{p^*}^j(h,\bar x_0), v_0:=x_0$.

By induction: $v_{j+1}:=\varphi_p(h,v_j)$, with $p=p_i$ and $jh\in[M_ih+kT_p,M_{i+1}h+kT_p)\cap[0,T]$.

Note that (see \eqref{e1}) $v_j=x(jh)$, so by \eqref{est1}, we have
\begin{equation}\label{est2}
|v_j-\bar v_j|\le ( |x_0-\bar x_0|+h\|A\|\|\bar x\|_0K) e^{(L+\|P\|_0\|A\|)T},
\end{equation}
when $jh\in [0,T]$. Next using \eqref{e4} and \eqref{order}, we derive
$$
\begin{gathered}
|u_{j+1}-v_{j+1}|=|\psi_p(h,u_j)-\varphi_p(h,v_j)|\\ \le|\psi_p(h,u_j)-\varphi_p(h,u_j)|
+|\varphi_p(h,u_j)-\varphi_p(h,v_j)|\\
\le |u_j-v_j| e^{(L+\|P\|_0\|A\|)h}+Mh^{r+1}.
\end{gathered}
$$
Then we get
\begin{equation}\label{est3}
\begin{gathered}
|u_j-v_j|\le Mh^{r+1}j e^{(L+\|p\|_0\|A\|)jh}
\le MTh^{r} e^{(L+\|P\|_0\|A\|)T},\, jh\in [0,T].
\end{gathered}
\end{equation}
Similarly, we derive
$$
\begin{gathered}
|\bar v_{j+1}^*-\bar v_{j+1}|=|\psi_{p^*}(h,\bar v_{j}^*)-\varphi_{p^*}(h,\bar v_j)|\\
\le|\psi_{p^*}(h,\bar v_{j}^*)-\varphi_{p^*}(h,\bar v_{j}^*)|
+|\varphi_{p^*}(h,\bar v_{j}^*)-\varphi_{p^*}(h,\bar v_j)|\\
\le |\bar v_{j}^*-\bar v_j| e^{(L+\|P\|_0\|A\|)h}+Mh^{r+1}.
\end{gathered}
$$
Then we get
\begin{equation}\label{est4}
|\bar v_{j}^*-\bar v_j|\le MTh^{r} e^{(L+\|P\|_0\|A\|)T},\, jh\in [0,T].
\end{equation}
Consequently, combining \eqref{est2}, \eqref{est3} and \eqref{est4}, we arrive at
\begin{equation}\label{est5}
\begin{gathered}
|\bar v_{j}^*-u_j|\le |\bar v_{j}^*-\bar v_j|+|\bar v_{j}-v_j|+|v_j-u_j|\\
\le h(|x_0-\bar x_0|+2MTh^{r-1}+\|A\|\|\bar x\|_0K)e^{(L+\|P\|_0\|A\|)T},\, jh\in [0,T].
\end{gathered}
\end{equation}

\begin{remark}
Inequality \eqref{est5} gives an estimation between numerical solutions of the switched equation \eqref{e1} and averaged equation \eqref{e2} by applying one-step method of order $r$ with step size $h$ and represents another generalization of the PS convergence for any utilized Runge-Kutta method (see \cite{Mao} where the convergence is proved via the standard RK method).
\end{remark}

\section{Lyapunov method approach}\label{sec5}

We consider \eqref{e1} on $I=[0,\infty)$ and assume that $f\in C^3(\R^n,\R^n)$. Motivated by \cite[pp. 168-169]{GH} and \cite{SVM}, we take a change of variables
\begin{equation}\label{e5}
\begin{gathered}
x(t)=y(t)+\Big(h(P_i-M_ip^*)+(t-h(M_i+kM_N))(p_{i+1}-p^*)\Big)Ay(t),
\end{gathered}
\end{equation}
for $t\in[M_ih+kT_p,M_{i+1}h+kT_p)$, where $k\in\Z$ and $P_i=\sum_{j=1}^ip_jm_j$, to derive
$$
\begin{gathered}
f\Big(y+\big(h(P_i-M_ip^*)
+(t-h(M_i+kM_N))(p_{i+1}-p^*)\big)By\Big)\\
+p_{i+1}A\Big(y+\big(h(P_i-M_ip^*)
+(t-h(M_i+kM_N))(p_{i+1}-p^*)\big)By\Big)\\
=f(x)+p_{i+1}Ax=\dot x
=\Big(\I+\big(h(P_i-M_ip^*)\\
+(t-h(M_i+kM_N))(p_{i+1}-p^*)\big)A\Big)\dot y
+(p_{i+1}-p^*)By.
\end{gathered}
$$
Next, since $P_i\le M_Np^*$, $p_i\le\|P\|_0$, $p^*\le\|P\|_0$ and $m_i\le M_i\le M_N$, we have
\begin{equation}\label{e6a}
\begin{gathered}
\|(h(P_i-M_ip^*)+(t-h(M_i+kM_N))(p_{i+1}-p^*))A\|\\
\le 2M_Nh\|A\|(p^*+\|P\|_0).
\end{gathered}
\end{equation}
Hence if
\begin{equation}\label{e6}
4M_Nh\|A\|(p^*+\|P\|_0)\le 1,
\end{equation}
then by Neumann theorem \cite{R}, the matrix
$$
\begin{gathered}
\I+\big(h(P_i-M_ip^*)
+(t-h(M_i+kM_N))(p_{i+1}-p^*)\big)A
\end{gathered}
$$
is invertible and
\begin{equation}\label{e7}
\begin{gathered}
\Big\|\Big(\I+\big(h(P_i-M_ip^*)
+(t-h(M_i+kM_N))(p_{i+1}-p^*)\big)A\Big)^{-1}\Big\|\le 2,\\
 \Big\|\Big(\I+\big(h(P_i-M_ip^*)
+(t-h(M_i+kM_N))(p_{i+1}-p^*)\big)A\Big)^{-1}-\I\Big\|\\
\le 4M_Nh\|A\|(p^*+\|P\|_0).
\end{gathered}
\end{equation}
Consequently, we get a system
\begin{equation}\label{e8}
\dot y(t)=f(y(t))+p^*Ay(t)+hg(y,t,h),
\end{equation}
where
$$
\begin{gathered}
g(y,t,h)=\frac{1}{h}\Big[\Big(\I+\big(h(P_i-M_ip^*)
+(t-h(M_i+kM_N))(p_{i+1}-p^*)\big)A\Big)^{-1}\\
\times\Big\{f\Big(y+\big(h(P_i-M_ip^*)
+(t-h(M_i+kM_N))(p_{i+1}-p^*)\big)By\Big)
+p_{i+1}A\big(h(P_i-M_ip^*)\\
+(t-h(M_i+kM_N))(p_{i+1}-p^*)\big)By\Big\}-f(y)\Big].
\end{gathered}
$$
By \eqref{e6a} and \eqref{e7}, we note
\begin{equation}\label{e9}
\begin{gathered}
|g(y,t,h)|\le \frac{1}{h}\Big\|\Big(\I+\big(h(P_i-M_ip^*)
+(t-h(M_i+kM_N))(p_{i+1}-p^*)\big)A\Big)^{-1}\Big\|\\
\times\Big\{\Big|f\Big(y+\big(h(P_i-M_ip^*)
+(t-h(M_i+kM_N))(p_{i+1}-p^*)\big)By\Big)-f(y)\Big|\\
+\|P\|_0\|A\||h(P_i-M_ip^*)
+(t-h(M_i+kM_N))(p_{i+1}-p^*)|\|A\||y|\Big\}\\
+\frac{1}{h}\Big\|\Big(\I+\big(h(P_i-M_ip^*)
+(t-h(M_i+kM_N))(p_{i+1}-p^*)\big)A\Big)^{-1}-\I\Big\||f(y)|\\
\le 4M_N\|A\|(p^*+\|P\|_0)
\times(L+\|p\|_0\|A\|+|f(0)|+L|y|).
\end{gathered}
\end{equation}
Assume that the averaged system \eqref{e2} has a strict Lyapunov function $V$ on a bounded subset $M\subset \R^n$ \cite{RHL}, so there is an $\al>0$ such that
\begin{equation}\label{e10}
\dot V(y):=V'(y)(f(y)+p^*By)\le-\al\quad\forall y\in M.
\end{equation}
According to \eqref{e9}, $\|g(y,t,\lambda)\|\le Q$ for any $y\in M$, $t$ and $h$ satisfying \eqref{e6} with
$$
\begin{gathered}
Q:=4P\|A\|(p^*+\|P\|_0)\Big(L+\|P\|_0\|A\|+|f(0)|
+L\textrm{dist}\{0,M\}\Big).
\end{gathered}
$$
Then if $y(t)\in M$, for a solution of (\ref{e8}), we get
$$
\begin{gathered}
\frac{d}{dt}V(y(t)):=hV'(y(t))\Big(f(y(t))+p^*Ay(t)
+\lambda g(y(t),t,h)\Big)\\
\le h(-\al+hQ)\le -h\al/2,
\end{gathered}
$$
when $h\le \frac{\al}{2Q}$ along with \eqref{e6}, i.e.
\begin{equation}\label{e11}
h\le \min\left\{\frac{\al}{2Q},\frac{1}{4M_N\|A\|(p^*+\|P\|_0)}\right\}.
\end{equation}
Therefore, by using the Lyapunov function, we can get some stability, dissipation and domain of attractions from (\ref{e2}) to (\ref{e1}). As a simple example, we can suppose that \eqref{e2} is inward oriented in some domain $\Omega\subset \R^n$ where $\Omega$ is an open bounded subset with a smooth boundary $\partial \Omega$, i.e. there is a smooth function $V :\R^n\to \R$ such that $V^{-1}(0)=\partial \Omega$ and \eqref{e10} holds for a positive constant $\al>0$ with $M=\partial \Omega$. Then, applying the above results, $\Omega$ is also inward for \eqref{e1} with $h$ satisfying \eqref{e11}, so $\Omega$ is invariant and locally attracting for \eqref{e1}. Then, the Brouwer fixed point theorem \cite{F1} ensures an existence of $T_p$-periodic solution of \eqref{e1} in $\Omega$.

\section{Conclusion}\label{sec6}

This paper is devoting to a comprehensive analytical study of the PS algorithm. We presented a general approach of his convergence and derive precise error estimates for solutions in terms of the switching step size. We also tackle the dynamical behavior of solutions via the Lyapunov method.

It is to note that the algorithm can be directly extended to the case when $p$ is not $T_p$-periodic.

Finally, as possible further directions, the convergence of the PS algorithm for discontinuous systems and for fractional-order systems will be considered.

\end{document}